\documentclass[pdflatex,sn-mathphys-num]{sn-jnl}

\usepackage{xcolor}
\usepackage{soul}
\sethlcolor{yellow}
\soulregister{\cite}{7}
\soulregister{\ref}{7}
\soulregister{\eqref}{7}
\soulregister{\label}{7}
\soulregister{\url}{7}

\usepackage{graphicx}
\usepackage{multirow}
\usepackage{amsmath,amssymb,amsfonts}
\usepackage{amsthm}
\usepackage{mathrsfs}
\usepackage{xcolor}
\usepackage{textcomp}
\usepackage{manyfoot}
\usepackage{booktabs}
\usepackage{placeins}
\usepackage{mathtools}
\usepackage{xurl}
\usepackage{array}
\usepackage{etoolbox}
\usepackage{adjustbox}
\hypersetup{hypertexnames=false}
\usepackage{float}

\makeatletter
\patchcmd{\@maketitle}{Corresponding author(s). E-mail(s): }{E-mail: }{}{\ClassError{sn-jnl-patch}{e-mail line patch failed}{}}
\patchcmd{\@maketitle}{\removelastskip\vskip24pt}{\removelastskip\vskip12pt}{}{}
\patchcmd{\@maketitle}{\removelastskip\vskip24pt}{\removelastskip\vskip12pt}{}{}
\newcount\authemailcnt
\def\email#1{\global\advance\emailcnt by 1\relax%
\if@corauemail
   \g@addto@macro\corrauthemail{\setcounter{footnote}{0}\textcolor{blue}{#1}}%
\else
   \ifnum\authemailcnt>0 \g@addto@macro\authemail{; }\fi
   \global\advance\authemailcnt by 1\relax
   \g@addto@macro\authemail{\setcounter{footnote}{0}\textcolor{blue}{#1}}%
\fi}
\makeatother

\theoremstyle{thmstyleone}%
\newtheorem{theorem}{Theorem}%
\newtheorem{proposition}{Proposition}%
\newtheorem{corollary}{Corollary}%

\AtBeginDocument{\raggedbottom}
\allowdisplaybreaks

\newcommand{\vE}{\mathbf{E}}
\newcommand{\vB}{\mathbf{B}}
\newcommand{\vS}{\mathbf{S}}
\newcommand{\vJ}{\mathbf{J}}
\newcommand{\vK}{\mathbf{K}}
\newcommand{\vA}{\mathbf{A}}
\newcommand{\vr}{\mathbf{r}}
\newcommand{\zero}{\mathbf{0}}
\newcommand{\nhat}{\hat{\mathbf{n}}}
\newcommand{\nc}{\hat{\mathbf{n}}_{c}}

\newcommand{\zhat}{\hat{\mathbf{z}}}
\newcommand{\xhat}{\hat{\mathbf{x}}}

\newcommand{\curl}{\nabla\times}
\newcommand{\dive}{\nabla\cdot}
\newcommand{\grad}{\nabla}
\newcommand{\gradt}{\nabla_{t}}
\newcommand{\pdt}[1]{\frac{\partial #1}{\partial t}}
\newcommand{\pdtt}[1]{\frac{\partial^{2} #1}{\partial t^{2}}}
\newcommand{\rt}{\rho_{\mathrm{tot}}}

\newcommand{\Jt}{\vJ_{\mathrm{tot}}}
\newcommand{\rC}{\rho_{C}}

\newcommand{\sigfs}{\sigma_{\mathrm{f},s}}
\newcommand{\Kfs}{\vK_{\mathrm{f},s}}

\newcommand{\eps}{\varepsilon_{0}}
\newcommand{\muz}{\mu_{0}}
\newcommand{\wme}{\omega^{2}\muz\eps}
\newcommand{\eiwt}{e^{-i\omega t}}

\begin{document}

\title[Scalar-Longitudinal Modes of Closed Cavities]{Scalar-Longitudinal Modes of Extended Electrodynamics in Perfectly Conducting Cavities and Trapped Modes of Divergence Cleaning}

\author*[1]{\fnm{N.} \sur{Rentzber}}\email{nrentzbe@uccs.edu}
\author[1]{\fnm{J.} \sur{Howard}}
\author[2,3]{\fnm{F. O.} \sur{Minotti}}\email{minotti@df.uba.ar}
\author[1]{\fnm{A.} \sur{Pinchuk}}
\author[4]{\fnm{G.} \sur{Modanese}}\email{giovanni.modanese@unibz.it}

\affil*[1]{\orgdiv{Department of Physics and Energy Science}, 

\orgname{University of Colorado}, \orgaddress{\city{Colorado Springs}, \postcode{80918}, \country{USA}}}

\affil[2]{\orgname{Universidad de Buenos Aires, Facultad de Ciencias Exactas y Naturales, Departamento de F\'{\i}sica}, \orgaddress{\city{Buenos Aires}, \country{Argentina}}}

\affil[3]{\orgname{CONICET-Universidad de Buenos Aires, Instituto de F\'{\i}sica Interdisciplinaria y Aplicada (INFINA)}, \orgaddress{\city{Buenos Aires}, \country{Argentina}}}

\affil[4]{\orgdiv{Faculty of Engineering}, 

\orgname{Free University of Bozen-Bolzano}, \orgaddress{\city{Bolzano}, \postcode{39100}, \country{Italy}}}

\abstract{\unboldmath In extended electrodynamics, the Lorenz condition is not imposed and ${C=\dive\vA+\muz\eps\,\partial\Phi/\partial t}$ is a dynamical field, which adds a scalar-longitudinal (SL) family to the TE, TM and TEM modes of a cavity. With $C=0$ on the wall of a closed perfectly conducting cavity, the classical resonances remain unchanged, while the SL modes are the Dirichlet eigenmodes of the volume, with $\vB=\zero$ and an electric field along $\grad C$. In uniform cavities they share the frequencies of the TM modes that vary along the axis, and for a connected wall the fundamental remains classical. Matching at a conductor described by its constitutive laws gives $\partial C/\partial n=0$ instead and couples the scalar-bearing modes to the Maxwell field. The results belong to the $C=0$ problem, the condition that separates the sectors and the one divergence cleaning imposes. The rectangular, circular cylinder and coaxial cavities provide benchmarks, together with the damping time that gives the slowest scalar mode its fastest decay.}


\maketitle

\clearpage
\setcounter{tocdepth}{2}
\tableofcontents
\clearpage

\section{Introduction}
\label{sec:intro}

In a perfectly conducting cavity of uniform cross-section, the classical source-free field is a superposition of transverse-electric (TE), transverse-magnetic (TM) and, on multiply connected cross-sections, transverse-electromagnetic (TEM) standing modes~\cite{jackson,pozar,collin}. Extended electrodynamics adds a scalar-longitudinal (SL) sector by retaining the scalar that the Lorenz gauge sets to zero and allowing it to evolve as a dynamical field~\cite{ohmura,aharonovbohm,vanvlaenderen,hively2012,modanese2017,minotti2021q}. The companion paper, Ref.~\cite{rentzberguide}, derives the extended system and solves it for a closed waveguide with a specified scalar wall condition. The formulation, its energy law and the wall condition at a physical conductor are those of Minotti and Modanese~\cite{minotti2021q,minotti2023}. There, the SL modes share the TM cutoffs. The present paper carries that analysis into closed cavities, treating resonances for general shapes and giving explicit solutions for uniform cavities by adding conducting end walls to the guides. Derivations that carry over unchanged are cited from the companion paper.
\\

The paper does not claim that the SL sector is realized in nature. Local charge conservation, whose failure is the only source of the scalar, is among the best-tested principles in physics~\cite{okun,borexino,majorana}. The aim is to state what the extended framework predicts in a theoretical closed conducting cavity, and what the same mathematics says about numerical solvers that carry such a scalar.
\\

For total charge density $\rt$ and total current density $\Jt$, the extended field equations derived in Section~1.3 of Ref.~\cite{rentzberguide} are

\begin{equation}
\begin{aligned}
\dive\vE &= \frac{\rt}{\eps}-\pdt{C}, &\qquad \curl\vE &= -\pdt{\vB},\\[2pt]
\dive\vB &= 0, &\qquad \curl\vB &= \muz\Jt+\muz\eps\,\pdt{\vE}+\grad C.
\end{aligned}
\label{eq:eed}
\end{equation}
\smallskip

\noindent Here $C\equiv\dive\vA+\muz\eps\,\partial\Phi/\partial t$ and $c=1/\sqrt{\muz\eps}$. Setting $C\equiv0$ recovers the classical Maxwell equations. The extended theory also allows a departure from local charge conservation, measured by $\Lambda$ through
\begin{equation*}
\pdt{\rt}+\dive\Jt=\Lambda.
\end{equation*}

\noindent Taking the divergence of the extended Amp\`ere's law and using the extended Gauss's law gives the scalar wave equation
\begin{equation}
\nabla^{2}C-\muz\eps\,\pdtt{C}=-\muz\Lambda.
\label{eq:Cwave}
\end{equation}
\smallskip

\noindent When charge is locally conserved, $\Lambda=0$ and the scalar wave equation is homogeneous. Such sources do not drive the scalar, although initial or boundary data can still produce a nonzero scalar field.
\\

The same scalar links the cavity problem to hyperbolic divergence cleaning, Section~11 of Ref.~\cite{rentzberguide}, whose auxiliary scalar couples to the fields exactly as $C$ does in Eq.~\eqref{eq:eed} and carries errors in Gauss's law as damped waves~\cite{munz2000,dedner2002}. In a closed conducting cavity those waves are trapped, and with no damping, a cleaning speed $c$ and a zero scalar wall value the trapped modes are the SL modes of the cavity.

\subsection{Results and Plan}
\label{sec:plan}

With $C=0$ on the wall $S$ of a cavity $V$, the SL modes are the eigenmodes of the Dirichlet Laplacian of the volume, with $\vB=\zero$, $\vE$ along $\grad C$ and frequencies $\omega_{j}^{\mathrm{SL}}$,
\begin{equation}
\boxed{\;
-\nabla^{2}\phi_{j}=\lambda_{j}\phi_{j}\ \ \text{in }V,
\qquad
\phi_{j}\big|_{S}=0,
\qquad
\omega_{j}^{\mathrm{SL}}=c\sqrt{\lambda_{j}}\, .
\;}
\label{eq:key}
\end{equation}
In a cavity of uniform cross-section, each SL mode shares its frequency with the TM mode of the same indices and $p\geq1$. When the wall is connected, at least three electromagnetic modes lie below the first SL mode. Charge-conserving sources leave the scalar sector empty if the scalar starts from zero. At a conductor described by its constitutive laws the wall condition is $\partial C/\partial n=0$. No separated SL mode exists in a uniform cavity under that condition, the modes that carry a scalar couple to the Maxwell field, and their study is left for separate work. In a cleaning solver with a zero scalar wall value, the trapped part separates exactly from the Maxwell part for any cleaning speed $v$ and damping time $\tau$. The slowest trapped mode decays fastest when $1/2\tau=v\sqrt{\lambda_{1}}$.
\\

Sections~\ref{sec:problem} to~\ref{sec:theorem} formulate the problem, decompose the field and establish the spectral theorems. Section~\ref{sec:canonical} gives the solutions for the three canonical cavities, and Section~\ref{sec:cleaning} applies the results to trapped modes of divergence cleaning.

\section{The Extended Cavity Problem}
\label{sec:problem}

\subsection{Wall Conditions and Classical Modes}
\label{sec:classicalmodes}

Let $V$ be a bounded, connected, vacuum-filled cavity with a piecewise smooth, perfectly conducting wall $S$ and outward unit normal $\nhat$. The wall may consist of several disjoint conductors. The electric field parallel to the wall and the magnetic field perpendicular to it both vanish,

\begin{equation}
\nhat\times\vE=\zero,
\qquad
\nhat\cdot\vB=0
\qquad\text{on }S .
\label{eq:pec}
\end{equation}
\smallskip

\noindent With $\nc=-\nhat$ directed from the conductor into the cavity, the wall sources are

\begin{equation}
\sigfs=\eps\,\vE\cdot\nc,
\qquad
\Kfs=\frac{1}{\muz}\,\nc\times\vB .
\label{eq:wallsources}
\end{equation}

For $\omega\neq0$, the source-free Maxwell equations with time dependence $\eiwt$ reduce to

\begin{equation}
\nabla^{2}\vE+\wme\,\vE=\zero,
\qquad
\dive\vE=0
\quad\text{in }V,
\qquad
\nhat\times\vE=\zero
\quad\text{on }S,
\label{eq:vhelm}
\end{equation}
\smallskip

\noindent with $\vB=-(i/\omega)\curl\vE$. This is an eigenvalue problem for $\wme$, with a discrete positive spectrum and finite multiplicities~\cite{monk}. For a cavity of uniform cross-section, write $V=A\times(0,L)$ with the axis along $z$. Its modes are standing combinations of the guide modes in Ref.~\cite{rentzberguide}, with axial wavenumber $k_{z}=p\pi/L$, written $k$ in Ref.~\cite{rentzberguide}, and transverse eigenvalue

\begin{equation}
\gamma^{2}=\wme-k_{z}^{2}
\label{eq:gamma}
\end{equation}
\smallskip

Each family is generated by a transverse profile $\psi(x,y)$ satisfying $(\gradt^{2}+\gamma^{2})\psi=0$ in $A$. For the TM family, $\psi=0$ on the side wall and the fields are

\begin{equation}
\begin{aligned}
E_{z} &= E_{0}\,\psi\cos(k_{z}z)\,\eiwt,\\[2pt]
\vE_{t} &= -E_{0}\,\frac{k_{z}}{\gamma^{2}}\,
\sin(k_{z}z)\,\gradt\psi\;\eiwt,\\[2pt]
B_{z} &= 0,\\[2pt]
\vB_{t} &= iE_{0}\,\frac{\omega\muz\eps}{\gamma^{2}}\,
\cos(k_{z}z)\,\zhat\times\gradt\psi\;\eiwt,
\qquad p\geq0.
\end{aligned}
\label{eq:tmrel}
\end{equation}
The end walls select the cosine dependence of $E_{z}$ and allow $p=0$. For the TE family, $\partial\psi/\partial n=0$ on the side wall and the fields are

\begin{equation}
\begin{aligned}
B_{z} &= B_{0}\,\psi\sin(k_{z}z)\,\eiwt,\\[2pt]
\vB_{t} &= B_{0}\,\frac{k_{z}}{\gamma^{2}}\,
\cos(k_{z}z)\,\gradt\psi\;\eiwt,\\[2pt]
E_{z} &= 0,\\[2pt]
\vE_{t} &= -iB_{0}\,\frac{\omega}{\gamma^{2}}\,
\sin(k_{z}z)\,\zhat\times\gradt\psi\;\eiwt,
\qquad p\geq1.
\end{aligned}
\label{eq:terel}
\end{equation}
The sine vanishes at both end walls, and $p=0$ gives no field. If the cross-section is bounded by $N$ disjoint conductors, it supports $N-1$ TEM families. Each is generated by a harmonic potential $\Phi_{0}(x,y)$ that is constant on each conductor, with $\gamma=0$ and fields

\begin{equation}
\begin{aligned}
\vE &= -E_{0}\,\sin(k_{z}z)\,\gradt\Phi_{0}\;\eiwt, &\qquad \wme &= k_{z}^{2},\\[2pt]
\vB &= iE_{0}\,\frac{k_{z}}{\omega}\,\cos(k_{z}z)\,\zhat\times\gradt\Phi_{0}\;\eiwt, &\qquad p &\geq 1 .
\end{aligned}
\label{eq:temrel}
\end{equation}
\smallskip

\noindent The end walls short the conductors together, so the coaxial cavity supports no static electric field. Including the static magnetic member of Eq.~\eqref{eq:temrel} at $p=0$ completes the three families~\cite{collin,jackson}.

\subsection{Energy and Uniqueness}
\label{sec:energy}

The uniqueness argument uses the quadratic identity derived in Section~3.1 of Ref.~\cite{rentzberguide} from Eq.~\eqref{eq:eed}. Its density is $u=\eps|\vE|^{2}/2+|\vB|^{2}/(2\muz)+C^{2}/(2\muz)$ and its flux is $\vS=(\vE\times\vB+C\vE)/\muz$, as in Eq.~(31) of Ref.~\cite{minotti2021q}. The physical energy density and flux follow from the Lagrangian of Ref.~\cite{minotti2021q} and depend on the potentials. The cavity problem specifies only $\vE$, $\vB$ and $C$, so it does not provide the potential data needed to determine those physical quantities. Here, the quadratic identity serves only as a positive conserved form.
\\

\begin{proposition}[Conserved Form and Uniqueness]
\label{prop:energy}
Let $(\vE,\vB,C)$ solve the source-free Eq.~\eqref{eq:eed} in $V$ with $\nhat\times\vE=\zero$ and $C=0$ on $S$. Then $U=\int_{V}u\,dV$ is constant, and the fields are fixed uniquely by their values at $t=0$.
\end{proposition}

\begin{proof}
The proof follows Proposition~1 of Ref.~\cite{rentzberguide}. The wall flux $\nhat\cdot\vS=[\vB\cdot(\nhat\times\vE)+C\,\nhat\cdot\vE]/\muz$ vanishes on $S$, and the cavity has no open ends. The difference of two solutions with the same initial fields is source-free and has $U=0$, so the positive definiteness of $u$ forces this difference to vanish.
\end{proof}

\subsection{Scalar Wall Condition}
\label{sec:extwall}

The classical conditions in Eq.~\eqref{eq:pec} leave some boundary data unspecified in the extended system, as the wave count in Section~3.2 of Ref.~\cite{rentzberguide} shows. For fields that vary only along $\nhat$, with $E_{n}=\nhat\cdot\vE$, the source-free Eq.~\eqref{eq:eed} gives

\begin{equation}
\pdt{}\big(E_{n}\mp cC\big)\mp c\,\frac{\partial}{\partial n}\big(E_{n}\mp cC\big)=0 ,
\label{eq:characteristic}
\end{equation}
\smallskip

\noindent so $E_{n}-cC$ travels into the cavity and requires a boundary condition~\cite{benzoni2007} that the usual conducting conditions do not supply. To make the remaining freedom explicit, choose a frequency $\omega$ for which $\wme$ is not a Dirichlet eigenvalue of $V$. Let $u$ solve $(\nabla^{2}+\wme)u=-\wme$ in $V$, with $u=0$ on $S$, and let $w$ be a constant. Then
\begin{equation}
C=w\,(1+u)\,\eiwt,
\qquad
\vE=-\frac{i\,w}{\omega\muz\eps}\,\grad u\;\eiwt,
\qquad
\vB=\zero
\label{eq:continuum}
\end{equation}
satisfies the source-free Eq.~\eqref{eq:eed} and both conditions of Eq.~\eqref{eq:pec}. With the conductor's conditions alone the cavity would ring at every frequency and select no spectrum. The third condition used in this paper is
\begin{equation}
C=0\qquad\text{on }S .
\label{eq:dirichletdata}
\end{equation}
\noindent Three facts support the choice in Eq.~\eqref{eq:dirichletdata}. First, the conductor requires it for every separated mode of the scalar sector alone in a uniform cavity, as Proposition~\ref{prop:wall} shows. Second, it is the condition used for the cleaning scalar at a perfect conductor in Refs.~\cite{munz2000,munz2000cpc,crestetto2012}. Third, it removes the quadratic wall flux $CE_{n}/\muz$ used in Proposition~\ref{prop:energy}.
\\

Equation~\eqref{eq:dirichletdata} specifies the boundary problem studied here. It is not the condition obtained at a physical conductor. Matching the solutions in the vacuum and in a conductor described by its constitutive laws, with continuity of the potentials and conservation of the surface charge, gives $\partial C/\partial n=0$ on the vacuum side with $C$ discontinuous across the surface, Eq.~(21) of Ref.~\cite{minotti2023} and Section~3.2 of Ref.~\cite{rentzberguide}. Under that condition a scalar plane wave at oblique incidence reflects together with a transverse wave, the scalar and Maxwell fields are coupled wherever $C$ varies along the wall, and the cavity modes for that case are a separate problem. Under Eq.~\eqref{eq:dirichletdata} a scalar plane wave incident at any angle $0\leq\theta<\pi/2$ reflects with $C_{r}=-C_{i}$ and no transverse wave, so the SL modes include waves that strike the wall at an angle. Equation~\eqref{eq:dirichletdata} is adopted here as the condition under which the scalar sector separates, Section~\ref{sec:decomp}, and as the condition imposed on the cleaning scalar at a perfect conductor, and the spectra below are results for that problem.
\\

At a physical conductor the wall condition follows from matching the solutions on the two sides of the surface, Section~3.2 of Ref.~\cite{rentzberguide}, where a thin conducting layer proposed by one of the authors (F.M.) as the realization of that matching is also described. The constitutive relation inside such a layer, Ohm's law or a Drude-type model, and the completion of the matched solution across the surface are left to a separate paper. None of the results below depend on them, since they belong to Eq.~\eqref{eq:dirichletdata}.

\section{Scalar-Longitudinal Sector}
\label{sec:sl}

\subsection{Reduction and Wall Condition}
\label{sec:reduction}

For a source-free field with time factor $\eiwt$, $\omega\neq0$ and no magnetic field, Faraday's law requires $\vE$ to be irrotational. The extended Amp\`ere's law gives $\grad C=i\omega\muz\eps\,\vE$ (Section~4.1 of Ref.~\cite{rentzberguide}), or equivalently,
\begin{equation}
\vE=-\frac{i}{\omega\muz\eps}\,\grad C,
\qquad
\vB=\zero .
\label{eq:slreduce}
\end{equation}
The extended Gauss's law, $\dive\vE=i\omega C$, then gives
\begin{equation}
\nabla^{2}C+\wme\,C=0.
\label{eq:helmholtzC}
\end{equation}
The electric field is a pure gradient, so the entire mode is generated by the single scalar $C$. By Eq.~\eqref{eq:slreduce}, the conductor's condition $\nhat\times\vE=\zero$ becomes $\nhat\times\grad C=\zero$ on $S$. This requires $C$ to be constant on each connected part of the wall but places no further restriction on its value. In a cavity of uniform cross-section, the geometry provides the remaining constraint.
\\

\begin{proposition}[Wall Condition on Separated SL Modes]
\label{prop:wall}
Let $V=A\times(0,L)$ and let $C=\psi(x,y)\,Z(z)\,\eiwt$ be a separated SL field with $\omega\neq0$. The conditions of Eq.~\eqref{eq:pec} hold if and only if $C=0$ on all of $S$. Then $\psi$ is a Dirichlet eigenfunction of the cross-section, $Z=\sin(p\pi z/L)$ with $p\geq1$, and no SL mode uniform along the axis exists.
\end{proposition}

\begin{proof}
On the side wall, the parallel components of $\grad C$ are $\psi\,dZ/dz$ and $Z\,\partial\psi/\partial s$, where $s$ is the arc length along the boundary of $A$. On the end walls, they are $Z(0)\gradt\psi$ and $Z(L)\gradt\psi$. If $Z(0)$ or $Z(L)$ were nonzero, $\psi$ would be constant. The side wall would then require $dZ/dz=0$, and Eq.~\eqref{eq:helmholtzC} would give $\omega=0$. Hence $Z=\sin(p\pi z/L)$ with $p\geq1$, and the side wall requires $\psi=0$ on the boundary of $A$. The converse follows immediately.
\end{proof}
This proposition is the standing-wave counterpart of Proposition~2 of Ref.~\cite{rentzberguide}, in which $E_{z}=(k_{z}/\omega\muz\eps)\,C$ on the side wall forces $C=0$ for $k_{z}\neq0$. It also corresponds to the cavity limit discussed in Section~4.3 of that reference. The end walls exclude the case $k_{z}=0$, so Eq.~\eqref{eq:dirichletdata} is required for the separated modes of the scalar sector alone in the cavities of Section~\ref{sec:canonical}. For a cavity of general shape, the conductor fixes only the derivative of $C$ along the wall, leaving Eq.~\eqref{eq:dirichletdata} as an independent condition.

\subsection{Wall Charge}
\label{sec:wallcharge}

The perpendicular electric field of an SL mode ends on wall charge, as discussed in Section~4.4 of Ref.~\cite{rentzberguide}. Equations~\eqref{eq:wallsources} and~\eqref{eq:slreduce} give
\begin{equation}
\sigfs=\frac{i}{\omega\muz}\,\frac{\partial C}{\partial n},
\qquad
\Kfs=\zero
\qquad\text{for every SL mode.}
\label{eq:slwallsources}
\end{equation}
Applying the divergence theorem together with Eq.~\eqref{eq:helmholtzC} gives the total wall charge
\begin{equation}
Q_{S}=\oint_{S}\sigfs\,dA=\frac{i}{\omega\muz}\int_{V}\nabla^{2}C\,dV=\eps\,\frac{d}{dt}\int_{V}C\,dV .
\label{eq:QS}
\end{equation}
A Maxwell mode has $\dive\vE=0$ and no net wall charge. An SL mode carries net wall charge whenever its profile has a nonzero mean, as in rectangular-cavity modes for which $n$, $m$ and $p$ are all odd. At each instant, its electric field is the one that would be produced inside a grounded conducting shell by the effective charge $\rC=-\eps\,\partial C/\partial t$ discussed in Section~11 of Ref.~\cite{rentzberguide}. Equation~\eqref{eq:QS} shows that this effective charge and the wall charge balance. Since $\Kfs=\zero$, no parallel current feeds the wall charge. Section~4.4 of Ref.~\cite{rentzberguide} discusses what would have to supply it, and shows that a conductor that conserves its surface charge and carries no current into its surface gives $\partial C/\partial n=0$ instead.

\subsection{Maxwell-SL Decomposition}
\label{sec:decomp}

\begin{proposition}[Maxwell and SL Parts of a Field]
\label{prop:decomp}
Let $V$ be any bounded cavity, and let $(\vE,\vB,C)$ be a source-free solution of Eq.~\eqref{eq:eed} with $\nhat\times\vE=\zero$ and $C=0$ on $S$. For each $t$ let $\Phi_{\mathrm{SL}}$ solve $\nabla^{2}\Phi_{\mathrm{SL}}=-\dive\vE$ in $V$ with $\Phi_{\mathrm{SL}}=0$ on $S$. Then
\begin{equation}
(\vE,\vB,C)=(\vE_{\mathrm{M}},\vB,0)+(\vE_{\mathrm{SL}},\zero,C),
\qquad
\vE_{\mathrm{SL}}=-\grad\Phi_{\mathrm{SL}},
\qquad
\grad C=-\muz\eps\,\pdt{\vE_{\mathrm{SL}}},
\label{eq:decomp}
\end{equation}
where the first part solves Maxwell's equations and the second solves Eq.~\eqref{eq:eed}, each with the perfect-conductor conditions. The two parts are orthogonal in the quadratic form, $\int_{V}\vE_{\mathrm{M}}\cdot\vE_{\mathrm{SL}}\,dV=0$, and each conserves its own contribution to $U$. At a fixed frequency $\omega\neq0$ the SL part is $\vE_{\mathrm{SL}}=-(i/\omega\muz\eps)\grad C$.
\end{proposition}

\begin{proof}
Define $g=C-\muz\eps\,\partial\Phi_{\mathrm{SL}}/\partial t$. The extended Gauss's law and Eq.~\eqref{eq:Cwave} with $\Lambda=0$ imply $\nabla^{2}g=0$. Since $g=0$ on $S$, it follows that $g\equiv0$ and $\grad C=-\muz\eps\,\partial\vE_{\mathrm{SL}}/\partial t$. The triple $(\vE_{\mathrm{SL}},\zero,C)$ then satisfies Eq.~\eqref{eq:eed}, while the remaining field has $\dive\vE_{\mathrm{M}}=0$ and no scalar. Both parts satisfy the wall conditions because $\Phi_{\mathrm{SL}}=0$ on $S$. Their orthogonality follows from integrating $\vE_{\mathrm{M}}\cdot\grad\Phi_{\mathrm{SL}}$ by parts. At a fixed frequency, $g=0$ becomes $\Phi_{\mathrm{SL}}=iC/(\omega\muz\eps)$, recovering Proposition~3 of Ref.~\cite{rentzberguide}.
\end{proof}
The decomposition holds for any time dependence, any cavity and any number of conductors in the wall, provided $C=0$ at the wall. Under $\partial C/\partial n=0$ the SL field of Eq.~\eqref{eq:slreduce} still satisfies the field equations, but it meets the conducting wall condition on its own only where $C$ is constant on each connected part of the wall, and the two parts are coupled wherever $C$ varies along the wall, Section~\ref{sec:extwall}.
\\

The decomposition is unique for oscillating fields. Any two decompositions of the same field would differ by a field with neither curl nor divergence that satisfies $\nhat\times\vE=\zero$ on the wall. Amp\`ere's law requires such a field to be time independent. This difference is therefore a static electric field between separate conductors of the wall, with $C=0$, and belongs to the Maxwell part.

\section{Cavity Spectral Theorems}
\label{sec:theorem}

\subsection{Dirichlet Spectrum}
\label{sec:dirichlet}

Let $-\nabla^{2}$ act on functions that vanish on $S$. On a bounded cavity with piecewise smooth wall the eigenvalue problem
\begin{equation}
\nabla^{2}\phi_{j}+\lambda_{j}\phi_{j}=0\quad\text{in }V,
\qquad
\phi_{j}=0\quad\text{on }S,
\qquad
\int_{V}\phi_{i}\phi_{j}\,dV=\delta_{ij},
\label{eq:dirichlet}
\end{equation}
has a discrete sequence of eigenvalues $0<\lambda_{1}<\lambda_{2}\leq\lambda_{3}\leq\cdots$ of finite multiplicity, with real eigenfunctions that are orthonormal and complete in $V$~\cite{mclean,couranthilbert}. Green's first identity gives $\int_{V}|\grad\phi_{j}|^{2}dV=\lambda_{j}$. The following theorem is the volume form of Theorem~1 of Ref.~\cite{rentzberguide}, with the cavity replacing the cross-section and resonance replacing cutoff.
\\

\begin{theorem}[Cavity Spectral Theorem]
\label{thm:spectrum}
Let $V$ be a bounded cavity with perfectly conducting wall $S$ and the condition $C=0$ on $S$. A source-free SL field with time factor $\eiwt$ exists if and only if $\wme=\lambda_{j}$ for some $j$, and it is then
\begin{equation}
C=A\,\phi_{j}\,e^{-i\omega_{j}t},
\qquad
\vE=-\frac{iA}{\omega_{j}\muz\eps}\,\grad\phi_{j}\,e^{-i\omega_{j}t},
\qquad
\vB=\zero,
\qquad
\omega_{j}=c\sqrt{\lambda_{j}} ,
\label{eq:slmode}
\end{equation}
with $\phi_{j}$ any eigenfunction that belongs to $\lambda_{j}$. At every other frequency $C\equiv0$ and every source-free field is a Maxwell field. The resonant frequencies of the extended cavity are those of the Maxwell cavity together with the $\omega_{j}$, and multiplicities add where the two sets meet.
\end{theorem}

\begin{proof}
For a source-free field with time factor $\eiwt$, Eq.~\eqref{eq:Cwave} and Eq.~\eqref{eq:dirichletdata} require $C$ to be either zero or a Dirichlet eigenfunction with $\wme=\lambda_{j}$. Proposition~\ref{prop:decomp} decomposes the field into a Maxwell part and the SL part of Eq.~\eqref{eq:slmode}, as in Theorem~1 of Ref.~\cite{rentzberguide}. Because this decomposition is unique, the multiplicities add.
\end{proof}

The initial-value problem is solved by expanding the scalar part of any source-free field in the eigenfunctions $\phi_{j}$. Each coefficient oscillates at $\omega_{j}$, with initial data set by $C$ and by $\partial C/\partial t=-\dive\vE$ at $t=0$, as in Corollary~1 of Ref.~\cite{rentzberguide}. The remaining field evolves as a Maxwell field.

\subsection{One Eigenproblem and the Fundamental}
\label{sec:oneproblem}

The extended cavity can also be posed as a single eigenproblem for the electric field. For the time factor $\eiwt$ the source-free Eq.~\eqref{eq:eed} gives $\curl\curl\vE=i\omega\grad C+\wme\,\vE$ and $\grad(\dive\vE)=i\omega\grad C$, and subtracting the two,
\begin{equation}
\nabla^{2}\vE+\wme\,\vE=\zero\quad\text{in }V,
\qquad
\nhat\times\vE=\zero,\quad\dive\vE=0\quad\text{on }S .
\label{eq:oneproblem}
\end{equation}
Compared with the Maxwell problem of Eq.~\eqref{eq:vhelm}, the divergence condition has moved from the volume to the wall, where it is the condition of Eq.~\eqref{eq:dirichletdata} written in terms of the field, since $C=-(i/\omega)\dive\vE$. In Theorem~1.1 of Ref.~\cite{costabeldauge}, Costabel and Dauge prove that the spectrum of their regularized Maxwell formulation splits into Maxwell and Dirichlet eigenvalues on simply connected Lipschitz domains. In convex polyhedra, this formulation coincides with their Lam\'e problem restricted to fields in $H^{1}$. In nonconvex polyhedra, that restriction excludes eigenfields without $H^{1}$ regularity, so the Lam\'e spectrum can lose part or all of the splitting. The field formulation used here does not impose that additional restriction. Theorem~\ref{thm:spectrum} identifies the Dirichlet gradient fields as the SL modes of the extended theory.
\\

\begin{proposition}[The Fundamental Is Electromagnetic]
\label{prop:rohleder}
Let the wall $S$ be connected. Then at least three electromagnetic modes, counted with multiplicity, lie strictly below the first SL mode, and at least $2k+1$ lie at or below the $k$-th SL mode.
\end{proposition}

\begin{proof}
Let $\alpha_{1}\leq\alpha_{2}\leq\cdots$ be the values of $\wme$ for the divergence-free solutions of Eq.~\eqref{eq:vhelm}, counted with multiplicity and with any static electric fields included at $\alpha=0$. Rohleder proved $\alpha_{3}<\lambda_{1}$ and $\alpha_{2k+1}\leq\lambda_{k}$ for every bounded connected cavity with a piecewise smooth wall free of cusps~\cite{rohleder2025}. A zero value occurs only when $S$ is disconnected. With $S$ connected every $\alpha_{i}$ is positive and belongs to a Maxwell mode.
\end{proof}

The count cannot be improved. In a cube of side $a$ exactly three modes, those with one vanishing index at $\omega a/c=\pi\sqrt{2}$, lie below the first SL mode at $\pi\sqrt{3}$. When the wall is disconnected the inequalities still hold, but some of the $\alpha_{i}$ then belong to static electric fields.

\subsection{Cavities of Uniform Cross-Section}
\label{sec:uniform}

\begin{theorem}[Coincidence Theorem]
\label{thm:rightcyl}
Let $V=A\times(0,L)$. Let $\gamma_{j}^{2}$ be the eigenvalues of $-\gradt^{2}$ on $A$ with $\psi=0$ on the boundary, let $\gamma_{j}'^{\,2}$ be its positive eigenvalues with $\partial\psi/\partial n=0$ there, and let $k_{z}=p\pi/L$. The four families of the cavity are
\begin{equation}
\begin{aligned}
\text{TM}\qquad & \wme=\gamma_{j}^{2}+k_{z}^{2}, & p&\geq0 ,\\
\text{TE}\qquad & \wme=\gamma_{j}'^{\,2}+k_{z}^{2}, & p&\geq1 ,\\
\text{TEM}\qquad & \wme=k_{z}^{2}, & p&\geq1 ,\\
\text{SL}\qquad & \wme=\gamma_{j}^{2}+k_{z}^{2}, & p&\geq1 ,
\end{aligned}
\label{eq:rightcyl}
\end{equation}
with one TEM family for each conductor of the cross-section beyond the first. The SL spectrum is the TM spectrum with $p\geq1$, mode for mode. The TM modes uniform along the axis, the TE modes and the TEM modes lie outside this pairing.
\end{theorem}

\begin{proof}
The first three lines are the standing-wave forms of the guide families described in Section~\ref{sec:classicalmodes}. For the fourth family, the Dirichlet eigenfunctions on the product domain have the form $\psi_{j}(x,y)\sin(p\pi z/L)$ with $p\geq1$. These functions are complete in $V$ and span every eigenspace, including the nonseparable superpositions allowed by a degenerate level. Proposition~\ref{prop:wall} shows that the conductor admits no other separated modes of the scalar sector alone.
\end{proof}
Section~4.3 of Ref.~\cite{rentzberguide} obtains the resonator modes of each transverse SL guide family at the allowed axial wavenumbers $k_{z}=p\pi/L$ with $p\geq1$. The theorem completes that description by including the other three families. Additional frequency coincidences are possible. A TM mode with $p=0$ can share an SL level, as TM$_{330}$ does with SL$_{114}$, SL$_{141}$ and SL$_{411}$ in a cube, while TE and SL frequencies coincide for all positive indices in a rectangular cavity. These coincidences depend on the geometry and do not follow from a general pairing rule. There is no scalar TEM branch because a Dirichlet eigenvalue is never zero. Frequencies alone therefore cannot reveal the scalar sector, although the paired modes can be distinguished by their fields. The SL mode has no magnetic field, whereas the TM mode has no scalar.
\\

This pairing follows from the uniform cross-section. In a spherical cavity of radius $R$, the condition $C=0$ on the wall becomes $j_{l}(\omega R/c)=0$. For $l\geq1$, this is also the condition for the TE modes of the sphere, whereas for $l=0$ it gives $\omega R/c=n\pi$, frequencies shared by no electromagnetic mode. In every cavity, the SL frequencies are the Dirichlet frequencies of the volume, but their positions relative to the Maxwell frequencies depend on the cavity shape.

\subsection{Sources}
\label{sec:excite}

Equation~\eqref{eq:Cwave} shows that the anomaly $\Lambda$ is the only source of the scalar, as discussed in Section~9 of Ref.~\cite{rentzberguide}. Expanding $C=\sum_{j}a_{j}(t)\,\phi_{j}$ gives
\begin{equation}
\frac{d^{2}a_{j}}{dt^{2}}+\omega_{j}^{2}\,a_{j}=\frac{\Lambda_{j}(t)}{\eps},
\qquad
\Lambda_{j}(t)=\int_{V}\phi_{j}\,\Lambda\,dV .
\label{eq:driven}
\end{equation}
At resonance, the lossless cavity has no steady state. A drive $\Lambda_{j}\cos(\omega_{j}t)$ switched on at $t=0$ produces $a_{j}=\Lambda_{j}\,t\sin(\omega_{j}t)/(2\eps\omega_{j})$, whose amplitude grows until losses outside the model limit it.
\\

\begin{corollary}[Dynamical Emptiness]
\label{cor:empty}
Let the matter inside the cavity conserve charge, $\Lambda\equiv0$, and let $C=\partial C/\partial t=0$ throughout $V$ at $t=0$. Then $C\equiv0$ for all $t$, whatever charges and currents are present, and the fields obey Maxwell's equations.
\end{corollary}

\begin{proof}
With $\Lambda=0$, Eq.~\eqref{eq:Cwave} is the homogeneous wave equation for $C$ alone, with $C=0$ on $S$ and zero initial data. The conserved quadratic integral $\int_{V}[\muz\eps(\partial C/\partial t)^{2}+|\grad C|^{2}]\,dV$ is initially zero, so $C\equiv0$.
\end{proof}
Searching for the scalar sector with a resonator therefore tests local charge conservation. An aperture, a port or an incoming scalar wave defines a different boundary-value problem, to which Corollary~\ref{cor:empty} does not apply.

\section{Canonical Cavities}
\label{sec:canonical}

This section derives the mode families of rectangular, circular and coaxial cavities from the canonical guides in Ref.~\cite{rentzberguide}. Each family has one generating field, $E_{z}$ for TM, $B_{z}$ for TE and $C$ for SL, and in a source-free cavity each generator $F$ obeys the Helmholtz equation
\begin{equation}
\nabla^{2}F+\wme\,F=0\qquad\text{in }V .
\label{eq:helm3d}
\end{equation}
In a cavity of uniform cross-section the trial solution $F=\psi(x,y)\,Z(z)$ separates it into
\begin{equation}
\gradt^{2}\psi+\gamma^{2}\psi=0,
\qquad
Z''+k_{z}^{2}Z=0,
\qquad
\gamma^{2}+k_{z}^{2}=\wme ,
\label{eq:helmsep}
\end{equation}
with a prime for the derivative with respect to the argument. For each geometry, the transverse profiles are those of the corresponding guide in Ref.~\cite{rentzberguide}, and the wall conditions determine the four families in turn. The remaining components follow from Eqs.~\eqref{eq:tmrel} to~\eqref{eq:temrel} for the Maxwell families and from Eq.~\eqref{eq:slreduce} for the SL family. As in Ref.~\cite{rentzberguide}, $n$ is the first transverse or azimuthal index, $m$ the second or radial index, and $p$ the axial index. Throughout $k_{z}=p\pi/L$, the factor $\eiwt$ is understood in every listing, and $E_{0}$, $B_{0}$ and $C_{0}$ are the amplitudes of the TM, TE and SL modes.

\subsection{Rectangular Cavity}
\label{sec:box}

Let $V=(0,a)\times(0,b)\times(0,L)$. The transverse profiles are the sine and cosine products of the rectangular guide described in Section~6.1 of Ref.~\cite{rentzberguide}, while the end walls select $\sin(k_{z}z)$ or $\cos(k_{z}z)$. Each wall condition eliminates either the cosine or the sine term in one factor and fixes
\begin{equation}
\begin{gathered}
k_{x}=\frac{n\pi}{a},
\qquad
k_{y}=\frac{m\pi}{b},
\qquad
k_{z}=\frac{p\pi}{L},
\qquad
\gamma_{nm}^{2}=k_{x}^{2}+k_{y}^{2},\\[2pt]
\omega_{nmp}^{2}\,\muz\eps=\gamma_{nm}^{2}+k_{z}^{2}=\pi^{2}\Big(\frac{n^{2}}{a^{2}}+\frac{m^{2}}{b^{2}}+\frac{p^{2}}{L^{2}}\Big).
\end{gathered}
\label{eq:boxk}
\end{equation}
The TM generator is $E_{z}=E_{0}\sin(n\pi x/a)\sin(m\pi y/b)\cos(p\pi z/L)$. The conditions $E_{z}=0$ on the side walls and $\partial E_{z}/\partial z=0$ on the end walls require $n,m\geq1$ and $p\geq0$.
The TE generator is $B_{z}=B_{0}\cos(n\pi x/a)\cos(m\pi y/b)\sin(p\pi z/L)$. The conditions $\partial B_{z}/\partial n=0$ on the side walls and $B_{z}=0$ on the end walls require $n,m\geq0$, $n+m\geq1$ and $p\geq1$.
The rectangular cross-section is bounded by a single conductor, so the cavity has no TEM mode.
\\

The SL generator is $C$, with $C=0$ on all six walls. On the side walls this is the TM condition. On the end walls it removes the cosine of $z$, where the TM condition removed the sine. The resulting profile is a product of three sines, with the wavenumbers of Eq.~\eqref{eq:boxk}. The normalized Dirichlet eigenfunctions in potential form are
\begin{equation}
\phi_{nmp}=\sqrt{\frac{8}{abL}}\,\sin\!\Big(\frac{n\pi x}{a}\Big)\sin\!\Big(\frac{m\pi y}{b}\Big)\sin\!\Big(\frac{p\pi z}{L}\Big),
\qquad n,m,p\geq1 ,
\label{eq:boxphi}
\end{equation}
and Eq.~\eqref{eq:slreduce} gives the complete SL$_{nmp}$ field,
\begin{equation}
\begin{aligned}
E_{x} &= -\frac{ik_{x}C_{0}}{\omega\muz\eps}\,\cos\!\Big(\frac{n\pi x}{a}\Big)\sin\!\Big(\frac{m\pi y}{b}\Big)\sin\!\Big(\frac{p\pi z}{L}\Big),\\[2pt]
E_{y} &= -\frac{ik_{y}C_{0}}{\omega\muz\eps}\,\sin\!\Big(\frac{n\pi x}{a}\Big)\cos\!\Big(\frac{m\pi y}{b}\Big)\sin\!\Big(\frac{p\pi z}{L}\Big),\\[2pt]
E_{z} &= -\frac{ik_{z}C_{0}}{\omega\muz\eps}\,\sin\!\Big(\frac{n\pi x}{a}\Big)\sin\!\Big(\frac{m\pi y}{b}\Big)\cos\!\Big(\frac{p\pi z}{L}\Big),\\[2pt]
B_{x} &= B_{y}=B_{z}=0,\\[2pt]
C &= C_{0}\sin\!\Big(\frac{n\pi x}{a}\Big)\sin\!\Big(\frac{m\pi y}{b}\Big)\sin\!\Big(\frac{p\pi z}{L}\Big),
\end{aligned}
\label{eq:boxsl}
\end{equation}
at the frequency of Eq.~\eqref{eq:boxk}. The electric field has the wall pattern of the TM mode with the same indices. On the wall $x=0$, where $\nc=\xhat$,
\begin{equation}
\sigfs(0,y,z)=-\frac{ik_{x}C_{0}}{\omega\muz}\,\sin\!\Big(\frac{m\pi y}{b}\Big)\sin\!\Big(\frac{p\pi z}{L}\Big),
\label{eq:boxsigma}
\end{equation}
and the other walls follow in the same way. The SL electric field is a gradient whose components are weighted by $k_{x}$, $k_{y}$ and $k_{z}$, whereas the TM field has zero divergence. The two fields therefore have the same sine and cosine factors but different relative component amplitudes. When $n$, $m$ and $p$ are all at least one, the TM, TE and SL modes share the same frequency, giving three modes in the extended theory compared with two in Maxwell theory.

\subsection{Circular Cavity}
\label{sec:cyl}

Consider the cylindrical cavity $r<a$, $0<z<L$ in cylindrical coordinates $(r,\varphi,z)$. Its transverse profiles, $J_{n}(\gamma r)\cos(n\varphi)$ or $J_{n}(\gamma r)\sin(n\varphi)$, are those of the circular guide described in Section~6.2 of Ref.~\cite{rentzberguide}. Here $x_{nm}$ is the $m$-th positive zero of $J_{n}$, and $x'_{nm}$ is the corresponding zero of $J_{n}'$. The modes below use the cosine branch. For $n\geq1$, the sine branch provides a second mode at every level. The Cartesian components follow from $E_{x}=E_{r}\cos\varphi-E_{\varphi}\sin\varphi$ and $E_{y}=E_{r}\sin\varphi+E_{\varphi}\cos\varphi$, with the same relations applying to $\vB$.
The TM generator is $E_{z}=E_{0}J_{n}(x_{nm}r/a)\cos(n\varphi)\cos(p\pi z/L)$, giving $\gamma_{nm}=x_{nm}/a$, $\omega_{nmp}=c\,(x_{nm}^{2}/a^{2}+k_{z}^{2})^{1/2}$ and $p\geq0$.
The TE generator is $B_{z}=B_{0}J_{n}(x'_{nm}r/a)\cos(n\varphi)\sin(p\pi z/L)$, giving $\gamma_{nm}'=x'_{nm}/a$, $\omega_{nmp}=c\,(x_{nm}'^{\,2}/a^{2}+k_{z}^{2})^{1/2}$ and $p\geq1$. Since $J_{0}'=-J_{1}$, the roots satisfy $x'_{0m}=x_{1m}$, so every TE$_{0mp}$ level coincides with the twofold TM$_{1mp}$ level.
The disk is bounded by a single conductor, so the circular cavity has no TEM mode.
\\

The SL generator is $C$, with $C=0$ on the curved wall and on both end walls. The curved wall gives the TM condition $J_{n}(\gamma_{nm} a)=0$ and $\gamma_{nm}=x_{nm}/a$, and the end walls use the sine along $z$. Equation~\eqref{eq:slreduce} gives the complete SL$_{nmp}$ field,
\begin{equation}
\begin{aligned}
C &= C_{0}\,J_{n}(\gamma_{nm} r)\cos(n\varphi)\sin\!\Big(\frac{p\pi z}{L}\Big),\\[2pt]
E_{r} &= -\frac{i\gamma_{nm}C_{0}}{\omega\muz\eps}\,J_{n}'(\gamma_{nm} r)\cos(n\varphi)\sin\!\Big(\frac{p\pi z}{L}\Big),\\[2pt]
E_{\varphi} &= \frac{inC_{0}}{\omega\muz\eps\,r}\,J_{n}(\gamma_{nm} r)\sin(n\varphi)\sin\!\Big(\frac{p\pi z}{L}\Big),\\[2pt]
E_{z} &= -\frac{ik_{z}C_{0}}{\omega\muz\eps}\,J_{n}(\gamma_{nm} r)\cos(n\varphi)\cos\!\Big(\frac{p\pi z}{L}\Big),\\[2pt]
B_{r} &= B_{\varphi}=B_{z}=0,
\end{aligned}
\label{eq:cylsl}
\end{equation}
with $\omega_{nmp}=c\,(x_{nm}^{2}/a^{2}+k_{z}^{2})^{1/2}$ and $p\geq1$. 
\\

TM modes with $p=0$ have no SL partner. These include TM$_{010}$, the lowest mode of a short cavity. At the shared level of TE$_{0mp}$ and TM$_{1mp}$, the two SL$_{1mp}$ modes increase the total number of modes from three to five.

\subsection{Coaxial Cavity}
\label{sec:coax}

Consider the region $a<r<b$, $0<z<L$ between two coaxial cylinders, closed by two end walls. The axis lies outside the cavity, so the radial functions are the Bessel cross products of the coaxial line described in Section~6.3 of Ref.~\cite{rentzberguide}. For TM and SL, the radial function vanishes on both cylinders,
\begin{equation}
\begin{aligned}
u_{nm}(r) &= J_{n}(\gamma_{nm} r)\,Y_{n}(\gamma_{nm} a)-Y_{n}(\gamma_{nm} r)\,J_{n}(\gamma_{nm} a),\\[2pt]
u_{nm}'(r) &= J_{n}'(\gamma_{nm} r)\,Y_{n}(\gamma_{nm} a)-Y_{n}'(\gamma_{nm} r)\,J_{n}(\gamma_{nm} a),
\end{aligned}
\label{eq:coaxuv}
\end{equation}
where $\gamma_{nm}=\chi_{nm}/a$ and $\chi_{nm}$ is the $m$-th positive root of the characteristic equation of Ref.~\cite{rentzberguide},
\begin{equation}
F_{n}\!\left(\chi;\frac{b}{a}\right)=J_{n}(\chi)\,Y_{n}\!\left(\frac{b\chi}{a}\right)-Y_{n}(\chi)\,J_{n}\!\left(\frac{b\chi}{a}\right)=0 .
\label{eq:coaxroots}
\end{equation}
Here $u_{nm}'$ denotes the radial derivative of $u_{nm}$ divided by $\gamma_{nm}$. For TE, the radial derivative vanishes on both cylinders, giving the cross product $v_{nm}$ of Ref.~\cite{rentzberguide}, with $\gamma_{nm}'=\chi_{nm}'/a$ and $\chi_{nm}'$ the roots of its TE characteristic equation. For the example used throughout, with $b=2a$ and $L=2b$, the values $\chi_{01}=3.1230$, $\chi_{11}=3.1966$ and $\chi_{21}=3.4069$ are close to the narrow-gap estimate $\pi a/(b-a)$. Similarly, $\chi_{11}'=0.6773$, $\chi_{21}'=1.3406$ and $\chi_{31}'=1.9789$ are close to $2na/(a+b)$. The equations below use the cosine branch, and the sine branch follows as in Section~\ref{sec:cyl}.
The TM generator is $E_{z}=E_{0}\,u_{nm}(r)\cos(n\varphi)\cos(p\pi z/L)$, with $\omega_{nmp}=c\,(\gamma_{nm}^{2}+k_{z}^{2})^{1/2}$ and $p\geq0$.
The TE generator is $B_{z}=B_{0}\,v_{nm}(r)\cos(n\varphi)\sin(p\pi z/L)$, with $\omega_{nmp}=c\,(\gamma_{nm}'^{\,2}+k_{z}^{2})^{1/2}$ and $p\geq1$. Since $J_{0}'=-J_{1}$ and $Y_{0}'=-Y_{1}$, the TE condition for $n=0$ coincides with the TM condition for $n=1$. Thus $\gamma_{0m}'=\gamma_{1m}$, and the TE$_{0mp}$ levels coincide with the TM$_{1mp}$ levels, as in the circular cavity.
The two cylinders support one TEM family. With $\gamma=0$, the potential obeys $\gradt^{2}\Phi=0$ in the ring, with $\Phi=V_{0}$ at $r=a$ and $\Phi=0$ at $r=b$, so $\Phi=V_{0}\ln(b/r)/\ln(b/a)$ as in Ref.~\cite{rentzberguide}. Equation~\eqref{eq:temrel} with $E_{0}\Phi_{0}=\Phi$ gives the complete field,
\begin{equation}
E_{r}=\frac{V_{0}}{r\ln(b/a)}\,\sin\!\Big(\frac{p\pi z}{L}\Big),\qquad B_{\varphi}=-\frac{iV_{0}}{c\,r\ln(b/a)}\,\cos\!\Big(\frac{p\pi z}{L}\Big),
\label{eq:coaxtem}
\end{equation}
and all other components are zero, with $\omega_{p}=p\pi c/L$ and $p\geq1$, set by the length alone. The first TEM mode lies below every TE mode. It is also the fundamental mode of the coaxial cavity whenever it lies below TM$_{010}$, which occurs for $L>\pi/\gamma_{01}$. The zero-frequency member of the family is the static field $B_{\varphi}\propto1/r$ of Section~\ref{sec:classicalmodes}.
\\

The SL generator $C$ vanishes on both cylinders, which are the TM conditions, so the radial function is again $u_{nm}$ with the roots of Eq.~\eqref{eq:coaxroots}, and the end walls select the sine along $z$. Equation~\eqref{eq:slreduce} gives the complete SL$_{nmp}$ field,
\begin{equation}
\begin{aligned}
E_{r} &= -\frac{iC_{0}\gamma_{nm}}{\omega\muz\eps}\,u_{nm}'(r)\cos(n\varphi)\sin\!\Big(\frac{p\pi z}{L}\Big),\\[2pt]
E_{\varphi} &= \frac{inC_{0}}{\omega\muz\eps\,r}\,u_{nm}(r)\sin(n\varphi)\sin\!\Big(\frac{p\pi z}{L}\Big),\\[2pt]
E_{z} &= -\frac{ik_{z}C_{0}}{\omega\muz\eps}\,u_{nm}(r)\cos(n\varphi)\cos\!\Big(\frac{p\pi z}{L}\Big),\\[2pt]
B_{r} &= B_{\varphi}=B_{z}=0,\\[2pt]
C &= C_{0}\,u_{nm}(r)\cos(n\varphi)\sin\!\Big(\frac{p\pi z}{L}\Big),
\end{aligned}
\label{eq:coaxsl}
\end{equation}
with $\omega_{nmp}=c\,(\gamma_{nm}^{2}+k_{z}^{2})^{1/2}$ and $p\geq1$. 
\\

The TEM family has no scalar partner. Its potential takes different constants on the two cylinders, and a scalar with that profile would violate Eq.~\eqref{eq:dirichletdata} and, by Proposition~\ref{prop:wall}, the conductor itself.

\subsection{Resonant Levels}
\label{sec:freq}

Table~\ref{tab:levels} lists the levels of the three cavities for every family, by index, starting from the lowest. The rectangular cavity has $(a,b,L)=(1,\,0.75,\,0.5)\,a$, the circular cavity has $L=2a$, and the coaxial cavity has $b=2a$ and $L=2b$. All entries are roots of the closed-form conditions above.
\begin{table}[!tbp]
\caption{Lowest resonant levels of the three canonical cavities by family and index $(n,m,p)$. Rectangular values are $\omega a/c$ for $b=0.75a$ and $L=0.5a$, circular values $\omega a/c$ for $L=2a$, and coaxial values $\omega b/c$ for $b=2a$ and $L=2b$. A blank means no mode in that family, SL and TM coincide where both exist, and SL modes require $p\geq1$. The lower right block lists the coaxial TEM levels and the lowest level of each family.}
\label{tab:levels}
\centering
{\setlength{\tabcolsep}{4.5pt}\renewcommand{\arraystretch}{0.92}
\begin{tabular}{@{}lccc@{\hspace{1.6em}}lccc@{}}
\toprule
\multicolumn{4}{@{}l}{\textit{Rectangular Cavity, $\omega a/c$}} & \multicolumn{4}{l@{}}{\textit{Circular Cavity, $\omega a/c$}}\\
\cmidrule(r{1.6em}){1-4}\cmidrule(l){5-8}
$(n,m,p)$ & TE & TM & SL & $(n,m,p)$ & TE & TM & SL\\
\midrule
$(1,1,0)$ & & 5.2360 & & $(0,1,0)$ & & 2.4048 & \\
$(1,0,1)$ & 7.0248 & & & $(1,1,1)$ & 2.4202 & 4.1412 & 4.1412\\
$(0,1,1)$ & 7.5514 & & & $(0,1,1)$ & 4.1412 & 2.8724 & 2.8724\\
$(2,1,0)$ & & 7.5514 & & $(2,1,1)$ & 3.4345 & 5.3705 & 5.3705\\
$(1,1,1)$ & 8.1789 & 8.1789 & 8.1789 & $(1,1,2)$ & 3.6414 & 4.9550 & 4.9550\\
$(2,0,1)$ & 8.8858 & & & $(1,1,0)$ & & 3.8317 & \\
$(1,2,0)$ & & 8.9473 & & $(0,1,2)$ & 4.9550 & 3.9564 & 3.9564\\
$(2,1,1)$ & 9.8236 & 9.8236 & 9.8236 & $(2,1,2)$ & 4.3815 & 6.0203 & 6.0203\\
$(3,1,0)$ & & 10.3137 & & $(3,1,1)$ & 4.4852 & 6.5707 & 6.5707\\
$(0,2,1)$ & 10.4720 & & & $(1,1,3)$ & 5.0593 & 6.0736 & 6.0736\\
$(2,2,0)$ & & 10.4720 & & $(2,1,0)$ & & 5.1356 & \\
$(1,2,1)$ & 10.9331 & 10.9331 & 10.9331 & $(0,1,3)$ & 6.0736 & 5.2905 & 5.2905\\
\midrule
\multicolumn{4}{@{}l}{\textit{Coaxial Cavity, $\omega b/c$}} & \multicolumn{4}{l@{}}{\textit{Coaxial TEM Levels, $\omega b/c$}}\\
\cmidrule(r{1.6em}){1-4}\cmidrule(l){5-8}
$(n,m,p)$ & TE & TM & SL & \multicolumn{1}{c}{$p$} & \multicolumn{3}{c}{$\omega b/c$}\\
\midrule
$(1,1,1)$ & 2.0743 & 6.5833 & 6.5833 & \multicolumn{1}{c}{1} & \multicolumn{3}{c}{1.5708}\\
$(2,1,1)$ & 3.1075 & 6.9926 & 6.9926 & \multicolumn{1}{c}{2} & \multicolumn{3}{c}{3.1416}\\
$(1,1,2)$ & 3.4212 & 7.1233 & 7.1233 & \multicolumn{1}{c}{3} & \multicolumn{3}{c}{4.7124}\\
$(2,1,2)$ & 4.1302 & 7.5032 & 7.5032 & \multicolumn{1}{c}{4} & \multicolumn{3}{c}{6.2832}\\
$(3,1,1)$ & 4.2581 & 7.6214 & 7.6214 & & & & \\
$(1,1,3)$ & 4.9032 & 7.9422 & 7.9422 & \multicolumn{4}{l@{}}{\textit{Lowest Level of Each Family}}\\
\cmidrule(l){5-8}
$(0,1,0)$ & & 6.2461 & & Cavity & TE & TM & SL\\
$(1,1,0)$ & & 6.3932 & & Rectangular & 7.0248 & 5.2360 & 8.1789\\
$(0,1,1)$ & 6.5833 & 6.4406 & 6.4406 & Circular & 2.4202 & 2.4048 & 2.8724\\
$(0,1,2)$ & 7.1233 & 6.9916 & 6.9916 & Coaxial & 2.0743 & 6.2461 & 6.4406\\
$(2,1,0)$ & & 6.8138 & & Coaxial TEM & \multicolumn{3}{c}{1.5708}\\
$(0,1,3)$ & 7.9422 & 7.8243 & 7.8243 & & & & \\
\botrule
\end{tabular}}
\end{table}

Four electromagnetic modes lie below the first SL level in the rectangular cavity. In the circular cavity the three modes below it are TM$_{010}$ and the twofold TE$_{111}$, as Proposition~\ref{prop:rohleder} guarantees. The coaxial cavity has $18$ electromagnetic levels below its first SL level, four TEM, twelve TE and the TM levels $010$ and $110$ that are uniform along the axis, or $31$ modes when the $\cos(n\varphi)$ and $\sin(n\varphi)$ partners are counted separately. A TM or SL profile must vanish on both cylinders and fit half a radial wavelength across the gap, which raises its frequency to about $\omega b/c=\pi b/(b-a)$, while a TE or TEM profile can vary slowly across the gap.

\subsection{Field Patterns}
\label{sec:illus}

Figure~\ref{fig:maps} compares the lowest TE mode with the lowest TM and SL pair in the circular cavity with $L=2a$. The TE mode lies below the first SL level. The TM and SL modes share a frequency, but their field patterns distinguish them. In the SL mode, the scalar reaches its maximum at the center and falls to zero on every wall. Its electric field oscillates parallel to $\grad C$, meets the walls perpendicularly and has no circulation. The wall charge has the same sign everywhere, in agreement with Eq.~\eqref{eq:QS}.
\FloatBarrier
\begin{figure}[!t]
\centering
\includegraphics[width=0.9\textwidth]{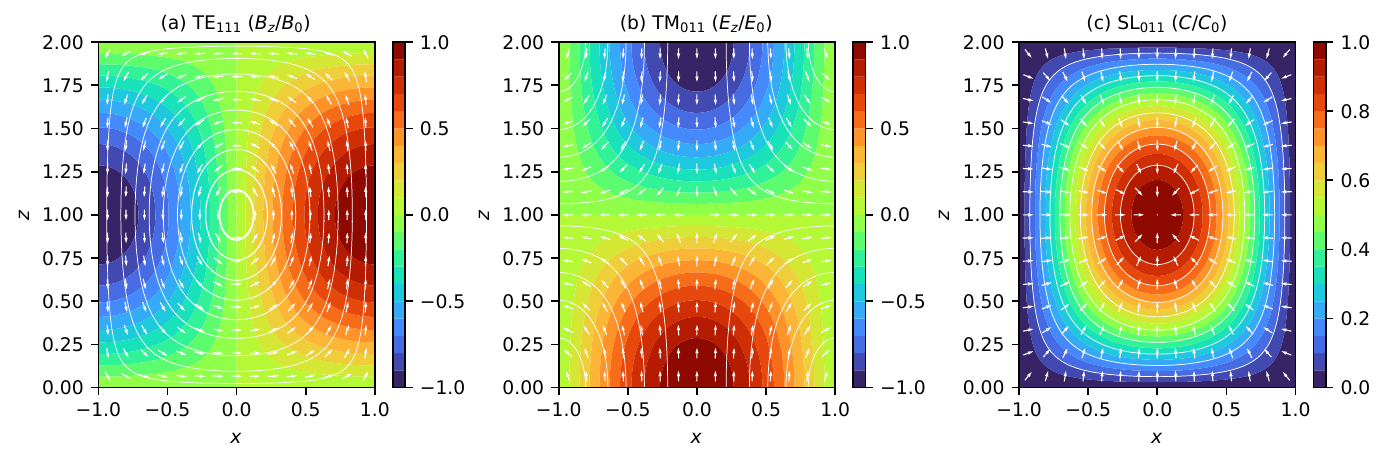}
\caption{The TE$_{111}$, TM$_{011}$ and SL$_{011}$ modes of the circular cavity with $L=2a$ in the plane $y=0$, with time and phase factors omitted and $B_0$, $E_0$ and $C_0$ the peak magnitudes of the profiles shown. The TE mode has frequency $\omega a/c=2.4202$, while the TM and SL modes share $\omega a/c=2.8724$. \textbf{(a)} TE$_{111}$, cosine angular branch. Colors show $B_z/B_0$, white curves and arrows the in-plane magnetic field, and the electric field is perpendicular to the plane. \textbf{(b)} TM$_{011}$. Colors show $E_z/E_0$, white curves and arrows the electric field, and $B_\varphi$ is perpendicular to the plane, a quarter period out of step with $\vE$ by the factor $i$ in Eq.~\eqref{eq:tmrel}. \textbf{(c)} SL$_{011}$. Colors and white level curves show $C/C_0$ and arrows $\grad C$, and the electric field $\vE=-(i/(\omega\muz\eps))\grad C$ is a quarter period out of step with $C$, while $\vB=\zero$.}
\label{fig:maps}
\end{figure}

\section{Trapped Modes of Divergence Cleaning}
\label{sec:cleaning}

In the cleaning system of Section~11 of Ref.~\cite{rentzberguide}, the auxiliary scalar $\psi$ takes the place of $C$ in Eq.~\eqref{eq:eed}. The cleaning speed $v$ enters through the term $v^{2}\grad\psi$ in the equation for $\partial\vE/\partial t$, replacing $c^{2}\grad C$, and the damping time $\tau$ through the term $-\psi/\tau$ in the equation for $\partial\psi/\partial t$. With $v=c$ and $\tau^{-1}=0$ the equations return to Eq.~\eqref{eq:eed}, with $\psi$ driven by the same anomaly $\Lambda$. The cleaning scalar is set to zero at a perfect conductor, as in Refs.~\cite{munz2000,munz2000cpc,crestetto2012}, so the wall condition is again Eq.~\eqref{eq:dirichletdata}, and the walls of a closed cavity trap the scalar inside.
\\

\begin{proposition}[Trapped Modes]
\label{prop:trapped}
Let $(\vE,\vB,\psi)$ solve the source-free cleaning system in a cavity $V$ with $\nhat\times\vE=\zero$ and $\psi=0$ on $S$, for any $v$ and any $\tau$. With $\Phi_{\mathrm{SL}}$ as in Proposition~\ref{prop:decomp},
\begin{equation}
(\vE,\vB,\psi)=(\vE_{\mathrm{M}},\vB,0)+(-\grad\Phi_{\mathrm{SL}},\zero,\psi),
\qquad
\pdt{}\,\grad\Phi_{\mathrm{SL}}=v^{2}\grad\psi .
\label{eq:trappedsplit}
\end{equation}
The first part solves Maxwell's equations and does not depend on $v$ or $\tau$. The second is the trapped part. With $\psi=\sum_{j}a_{j}(t)\,\phi_{j}$ its amplitudes are damped oscillators,
\begin{equation}
\frac{d^{2}a_{j}}{dt^{2}}+\frac{1}{\tau}\,\frac{da_{j}}{dt}+v^{2}\lambda_{j}\,a_{j}=0 ,
\label{eq:trappedode}
\end{equation}
\end{proposition}

\begin{proof}
The argument follows Proposition~\ref{prop:decomp}, with $g=v^{2}\psi-\partial\Phi_{\mathrm{SL}}/\partial t$. Taking the divergence of the cleaning equation for $\partial\vE/\partial t$ gives $\nabla^{2}g=0$ with $g=0$ on $S$. Thus $g\equiv0$, which gives the relation in Eq.~\eqref{eq:trappedsplit}. Taking the time derivative of the equation for $\partial\psi/\partial t$ then gives $\partial^{2}\psi/\partial t^{2}+\tau^{-1}\partial\psi/\partial t=v^{2}\nabla^{2}\psi$. Projecting this equation onto $\phi_{j}$ yields Eq.~\eqref{eq:trappedode}.
\end{proof}
In the continuous problem with this wall condition, cleaning therefore never touches the physical field in a closed conducting cavity, and a scheme preserves this only with the compatible discrete split of Section~\ref{sec:damping}, acting on the Dirichlet modes of the volume alone. A mode with $v^{2}\lambda_{j}>1/4\tau^{2}$ rings at $\Omega_{j}=(v^{2}\lambda_{j}-1/4\tau^{2})^{1/2}$ and decays at the rate $1/2\tau$, and a mode below that threshold is overdamped. For a Dirichlet eigenpair $(\lambda_{j},\phi_{j})$ and an amplitude $\psi_{0}$,
\begin{equation}
\begin{gathered}
\psi = \psi_{0}\,\phi_{j}(\vr)\,e^{-t/2\tau}\cos(\Omega_{j}t),\\[2pt]
\Omega_{j}^{2}=v^{2}\lambda_{j}-\frac{1}{4\tau^{2}},\\[2pt]
\vE = \frac{\psi_{0}}{\lambda_{j}}\,e^{-t/2\tau}\Big[\frac{1}{2\tau}\cos(\Omega_{j}t)-\Omega_{j}\sin(\Omega_{j}t)\Big]\grad\phi_{j}(\vr),\\[2pt]
\vB = \zero,
\end{gathered}
\label{eq:benchmark}
\end{equation}
solves the source-free cleaning system exactly, for any cavity and any $v$ and $\tau$ satisfying $v^{2}\lambda_{j}>1/4\tau^{2}$. It starts from $\psi=\psi_{0}\phi_{j}$, $\vE=(\psi_{0}/2\tau\lambda_{j})\grad\phi_{j}$ and $\vB=\zero$, and a code started from these data must ring at $\Omega_{j}$ and decay at the rate $1/2\tau$. Below this threshold, replace $\cos(\Omega_{j}t)$ by $\cosh(\kappa_{j}t)$ and $\Omega_{j}\sin(\Omega_{j}t)$ by $-\kappa_{j}\sinh(\kappa_{j}t)$, with $\kappa_{j}^{2}=1/4\tau^{2}-v^{2}\lambda_{j}$. The mode then decays without ringing, with a long-time decay rate of $1/2\tau-\kappa_{j}$. At the threshold itself the amplitude is simply $e^{-t/2\tau}$ for these initial data, since the term proportional to $t$ that critical damping allows in general has zero coefficient here. Section~\ref{sec:canonical} makes the benchmark explicit for the three cavities, and the uniform-cavity form reduces to that of Ref.~\cite{rentzberguide}.
\\

At $v=c$ without damping the trapped frequencies are the SL frequencies, so in a uniform cavity each sits on a TM frequency and only the fields tell constraint error from physics. Changing $v$ multiplies the SL frequencies in Table~\ref{tab:levels} by $v/c$ while the Maxwell frequencies stay fixed, so the trapped modes move away from their TM partners.

\subsection{Damping}
\label{sec:damping}

Each trapped mode is a damped oscillator, by Eq.~\eqref{eq:trappedode}, and the error as a whole disappears only as fast as its slowest mode, which is the lowest one. Shortening the damping time $\tau$ helps only up to a point. Once the lowest mode is critically damped, a still shorter $\tau$ overdamps it, and an overdamped mode creeps back to zero at the smaller rate $1/2\tau-(1/4\tau^{2}-v^{2}\lambda_{1})^{1/2}$. The best choice is therefore critical damping of the lowest mode,
\begin{equation}
\frac{1}{2\tau_{*}}=v\sqrt{\lambda_{1}} ,
\label{eq:bestrate}
\end{equation}
so the best decay rate is the first SL frequency of the cavity, evaluated at the cleaning speed. Equation~\eqref{eq:bestrate} selects the damping time from the slowest trapped mode, the worst case over the scalar data the code may hold. It optimizes the long-time decay rate and gives no bound on the error at a particular instant, and when the lowest mode is absent the slowest populated mode sets the rate. No uniform $\tau$ gives this slowest mode a faster exponential decay than $e^{-v\sqrt{\lambda_{1}}\,t}$. Once that rate is reached, faster decay requires increasing $v$. In the notation of Dedner et al.~\cite{dedner2002}, with cleaning speed $c_{h}$ and damping rate $c_{h}^{2}/c_{p}^{2}$, Eq.~\eqref{eq:bestrate} reads $c_{p}^{2}/c_{h}=1/(2\sqrt{\lambda_{1}})$, a length set by the cavity and not by the mesh, about $0.092\,a$ for a cube of side $a$.
\\

A charge-conservation error in the scheme acts as a source $\Lambda$, adding the term $\Lambda_{j}(t)/\eps$ to Eq.~\eqref{eq:trappedode} as in Eq.~\eqref{eq:driven}. When the error drives a mode at its undamped frequency $\omega=v\sqrt{\lambda_{j}}$, damping limits the response that grows without bound in Section~\ref{sec:excite}, and the amplitude settles at $|\Lambda_{j}|\tau/(\eps\omega)$, proportional to $\tau$. Weak damping lets a large error build up near a trapped frequency, while strong damping makes the lowest modes overdamped and slow to decay. The trapped modes remain unexcited if the scheme preserves the split of Proposition~\ref{prop:trapped}, with $\dive\curl=0$ on the grid, charge-conserving sources and no trapped part initially. Then $\Lambda=0$ throughout the evolution, giving the discrete counterpart of Corollary~\ref{cor:empty}.

\section{Discussion and Conclusions}
\label{sec:disc}
\subsection{What the Results Say}

With $C=0$ on the wall, a closed perfectly conducting cavity carries a family of SL resonances at the Dirichlet frequencies $c\sqrt{\lambda_{j}}$ of the volume, and the classical modes remain unchanged. The Maxwell and SL fields evolve independently. In a cavity of uniform cross-section each SL mode shares its frequency with the TM mode of the same indices that varies along the axis, and whatever the shape the fundamental mode is electromagnetic when the wall is connected.
\\

All of this holds for the condition $C=0$. At a conductor described by its constitutive laws, Section~\ref{sec:extwall}, the matching gives $\partial C/\partial n=0$ instead. Wherever the scalar varies along the wall the scalar and Maxwell fields are coupled, no separated SL mode exists in a uniform cavity, and those coupled modes are left for a separate analysis.
\\

The same equations describe the auxiliary scalar of divergence cleaning in a closed cavity, where the wall condition is part of the numerical formulation and no model of the conductor is needed, and for any cleaning speed and damping time the trapped scalar separates from the Maxwell field. By Corollary~\ref{cor:empty} a charge-conserving source leaves the SL modes unexcited when the scalar is absent initially, which is why ordinary cavities show no sign of them. A solver can be checked against the fields of Section~\ref{sec:canonical}, the levels of Table~\ref{tab:levels} and the exact solution of Eq.~\eqref{eq:benchmark}, and the spectrum fixes the damping time of Eq.~\eqref{eq:bestrate} that gives the slowest trapped mode its fastest long-time decay.

\subsection{Relation to Prior Work}

Rohleder's inequalities between the Maxwell and Dirichlet eigenvalues fix the position of the fundamental, and the decomposition of Costabel and Dauge gives the same separation of the spectrum from a regularized Maxwell problem, Section~\ref{sec:oneproblem}~\cite{rohleder2025,costabeldauge}.
\\

In hyperbolic cleaning~\cite{munz2000,munz2000cpc,dedner2002}, Fornet and Mouysset, building on their earlier analysis~\cite{fornet2012}, show why the initial data matter as well as the wall condition, with a rectangular-guide example in which a zero scalar wall value still admits non-Maxwell waves from nonzero scalar initial data, and an extension that recovers the Maxwell solutions for compatible data and charge-conserving sources~\cite{fornet2022}. Corollary~\ref{cor:empty} is the corresponding statement for the wall condition used here.
\\

The material-interface problem of Minotti and Modanese~\cite{minotti2023} determines the wall condition at a physical conductor, $\partial C/\partial n=0$ on the vacuum side, Section~\ref{sec:extwall}, and the thin layer proposed by one of the authors as its realization is noted there and left to a separate paper.

\newpage

\backmatter

\section*{Statements and Declarations}
\noindent\textbf{Funding}\quad No funding was received for this study.

\medskip
\noindent\textbf{Competing Interests}\quad The authors declare no competing interests.

\medskip
\noindent\textbf{Data Availability}\quad The data generated in this study are the tabulated values in the manuscript, which follow from the equations given.

\medskip
\noindent\textbf{Code Availability}\quad A Python script was used to compute modes and prepare figures. The definitions, derivations, inputs, and conclusions are fully stated in the manuscript and do not depend on access to the script. The Python script is however available upon request.

\smallskip
\noindent\textbf{Author Contributions}\quad N.R. formulated the boundary-value problem, carried out the derivations and computations, and wrote the manuscript. F.M. and G.M. developed the formulation of extended electrodynamics in material media and its energy law on which this work is based, derived the scalar wall condition at a physical conductor, and revised the manuscript. J.H. contributed to the development of the manuscript. A.P. reviewed and analyzed the manuscript.

\FloatBarrier

\end{document}